\documentclass[a4paper,UKenglish,cleveref, autoref, thm-restate]{lipics-v2021}
\usepackage{booktabs}
\graphicspath{{ISAAC_2021/figures/}}%helpful if your graphic files are in another directory

\title{Dynamic Data Structures for $k$-Nearest Neighbor Queries} %TODO Please add

\author{Sarita de Berg}{Department of Information and Computing Sciences, Utrecht University, The Netherlands}{s.deberg@uu.nl}{}{}%TODO mandatory, please use full name; only 1 author per \author macro; first two parameters are mandatory, other parameters can be empty. Please provide at least the name of the affiliation and the country. The full address is optional

\author{Frank Staals}{Department of Information and Computing Sciences, Utrecht University, The Netherlands}{f.staals@uu.nl}{}{}

\authorrunning{S. de Berg and F. Staals} %TODO mandatory. First: Use abbreviated first/middle names. Second (only in severe cases): Use first author plus 'et al.'

\Copyright{Sarita de Berg and Frank Staals} %TODO mandatory, please use full first names. LIPIcs license is "CC-BY";  http://creativecommons.org/licenses/by/3.0/

 %Please choose ACM 2012 classifications from https://dl.acm.org/ccs/ccs_flat.cfm 

  %TODO mandatory; please add comma-separated list of keywords

\category{} %optional, e.g. invited paper

\relatedversion{} %optional, e.g. full version hosted on arXiv, HAL, or other respository/website
\nolinenumbers %uncomment to disable line numbering

\hideLIPIcs  %uncomment to remove references to LIPIcs series (logo, DOI, ...), e.g. when preparing a pre-final version to be uploaded to arXiv or another public repository

\EventEditors{John Q. Open and Joan R. Access}
\EventNoEds{2}
\EventLongTitle{42nd Conference on Very Important Topics (CVIT 2016)}
\EventShortTitle{CVIT 2016}
\EventAcronym{CVIT}
\EventYear{2016}
\EventDate{December 24--27, 2016}
\EventLocation{Little Whinging, United Kingdom}
\EventLogo{}
\SeriesVolume{42}
\ArticleNo{23}
\usepackage{xspace}

\graphicspath{ {figures/} }

\newcommand {\R} {\mathbb {R}}

\newcommand{\mkmcal}[1]{\ensuremath{\mathcal{#1}}\xspace}
\newcommand{\D}{\mkmcal{D}}

\newcommand{\T}{\mkmcal{T}}
\renewcommand{\P}{\mkmcal{P}}

\newcommand{\etal}{et al.\xspace}

\newcommand{\Fi}{F^{(i)}}
\newcommand{\Fx}[1]{F^{(#1)}}
\newcommand{\Fbad}[1]{F_{\text{bad}}^{(#1)}}
\newcommand{\Flive}[1]{F_{\text{live}}^{(#1)}}
\newcommand{\Ti}{\mkmcal{T}^{(i)}}
\newcommand{\tower}[1]{\mkmcal{T}^{(#1)}}

\newcommand{\eps}{\ensuremath{\varepsilon}\xspace}

\def\polylog{\operatorname{polylog}}

\begin{document}

\maketitle

\begin{abstract}
Our aim is to develop dynamic data structures that support $k$-nearest neighbors ($k$-NN) queries for a set of $n$ point sites in the plane in $O(f(n) + k)$ time, where $f(n)$ is some polylogarithmic function of $n$. The key component is a general query algorithm that allows us to find the $k$-NN spread over $t$ substructures simultaneously, thus reducing an $O(tk)$ term in the query time to $O(k)$. Combining this technique with the logarithmic method allows us to turn 
any static $k$-NN data structure into a data structure supporting both
efficient insertions and queries. For the fully dynamic case, this
technique allows us to recover the deterministic, worst-case,
$O(\log^2n/\log\log n +k)$ query time for the Euclidean distance
claimed before, while preserving the polylogarithmic update times. We
adapt this data structure to also support fully dynamic
\emph{geodesic} $k$-NN queries among a set of sites in a simple
polygon. For this purpose, we design a shallow cutting based,
deletion-only $k$-NN data structure. More generally, we obtain a
dynamic planar $k$-NN data structure for any type of distance functions for which we can build vertical shallow cuttings. We apply all of our methods in the plane for the Euclidean distance, the geodesic distance, and general, constant-complexity, algebraic distance functions.
\end{abstract}

\section{Introduction}

In the $k$-nearest neighbors ($k$-NN) problem, we are given a set of
$n$ point sites $S$ in some domain, and we wish to preprocess these
points such that given a query point $q$ and an integer $k$, we can find
the $k$ sites in $S$ `closest' to $q$ efficiently. This static problem
has been studied in many different
settings~\cite{Andoni08,Chan00,Chan16,lee1982kthorder_vd,LiuJournal}. We study the dynamic version of the planar $k$-nearest neighbors
problem, in which the set of sites $S$ may be subject to updates,
i.e. insertions and deletions. We are particularly interested in two
settings: (i) a setting in which the domain containing the sites
contains (polygonal) obstacles, and in which we measure the distance
between two points by their geodesic distance, the length of the
shortest obstacle avoiding path, and (ii) a setting in which only
insertions into $S$ are allowed (i.e. no deletions).

In many applications involving distances and shortest paths, the
entities involved cannot travel towards their destination in a straight
line. For example, a person walking through the city center may want
to find the $k$ closest restaurants that currently have seats
available. However, he or she cannot walk through walls, and hence, we
need to explicitly incorporate such obstacles into the problem. This
introduces additional complications as a single shortest path in a
polygon with $m$ vertices may have complexity $\Theta(m)$, and thus it
requires $\Theta(m)$ time to compute such a path. We wish to limit the
dependency on $m$ in the space, query, and update times of our data
structure as much as possible. In particular, we want to avoid having
to spend $\Omega(m)$ time and space when we insert or delete a new
site in $S$. In terms of the above example, we wish to avoid having to
spend $\Omega(m)$ time every time some seats open up causing a
restaurant to become available.

The second setting is motivated by classification
problems. In $k$-nearest neighbor classifiers the sites in $S$ all
have a label, and the label of some query point $q$ is predicted based
on the labels of the $k$ sites nearest to
$q$~\cite{cover1967nearest}. When this label turns out to be
sufficiently accurate, it is customary to then extend the data
set by adding $q$ to $S$. Hence, this naturally leads to the question
whether there is an insertion-only data structure that can efficiently
answer $k$-nearest neighbor queries.

\subparagraph{The static problem.} If the set of sites $S$ is static,
and $k$ is known a priori, one option is to build the (geodesic)
$k^\mathrm{th}$-order Voronoi diagram of $S$~\cite{Liu13}. This yields very fast $O(\log (n+m) + k)$ query
times, where $m$ is the complexity of the domain $\P$, however it is
costly in space, as even in a simple polygon the diagram has size
$O(k(n-k) + km)$. Moreover, $k$ needs to be known a priori. In the
scenario where the domain is the Euclidean plane, much more space
efficient solutions have been developed. There is an optimal linear
space data structure achieving $O(\log n + k)$ query time after
$O(n\log n)$ deterministic preprocessing
time~\cite{Afshani09,Chan16}. Very recently, Liu showed how to achieve
the same query and expected preprocessing time for general constant-complexity distance functions
for arbitrary sites in $\R^2$, using $O(n\log\log n)$ space
\cite{LiuJournal}. In
case $\P$ is a simple polygon with $m$ vertices, the problem has not
explicitly been studied. The only known solution using less space than
just storing the $k^\mathrm{th}$-order Voronoi diagram is the fully
dynamic $1$-NN structure of Agarwal \etal~\cite{Staals18}. It uses
$O(n\log^3 n\log m + m)$ space, and answers queries in
$O(k\polylog(n+m))$ time (by deleting and reinserting the $k$-closest
sites to answer a query).

\subparagraph{Issues when inserting sites.} Since nearest
neighbor searching is decomposable, we can apply the logarithmic
method~\cite{Overmars83} to turn a static $k$-NN searching data
structure into an insertion-only data structure. For example, in the
Euclidean plane this yields a linear space data structure with $O(\log^2 n)$
insertion time. However, since this partitions the set of sites $S$
into $O(\log n)$ subsets, and we do not know how many of the
$k$-nearest sites appear in each subset, we may have to consider up to
$k$ sites from \emph{each} of the subsets, thus yielding an
$O(k\log n)$ term in the query time. In
Section~\ref{sec:query_procedure} we will present a general technique
that allows us to avoid this additional $O(\log n)$ factor.

\subparagraph{Fully dynamic data structures.} In case we wish to
support both insertions and deletions, the problem becomes more
complicated, and the existing solutions thereby much more involved. When we
again consider the plane, and we wish to report only one nearest
neighbor (i.e. $1$-NN searching), several efficient fully dynamic data
structures exist~\cite{Chan10,Chan19,Kaplan17}. Actually, all these
data structures are variants of the same data structure by Chan
\cite{Chan10}. For the Euclidean distance, the current best result
using linear space achieves $O(\log^2 n)$ worst-case query time,
$O(\log^2 n)$ insertion time, and $O(\log^4 n)$ deletion
time~\cite{Chan19}. These results are deterministic, and the update
times are amortized. The variant by Kaplan \etal~\cite{Kaplan17}
achieves similar results for general distance functions: $O(n\log^3 n)$ space, $O(\log^2 n)$ worst-case query time, and expected $O(\polylog n)$ amortized update time. Using recent results on shallow
cuttings by Liu~\cite{LiuJournal} the space can be reduced to
$O(n\log n)$, the insertion time to $O(\log^2n)$, and the deletion time to $O(\log^4 n)$. The update times remain expected amortized. These data structures can also be used to answer $k$-NN
queries, but when answering a query we run into the same
problem as in the insertion-only case. That is, we consider $O(\log n)$ subsets, which results in a query time of
$O(\log^2 n + k\log n)$~\cite{Chan10, LiuJournal}.

For the Euclidean case, Chan argues that the above data structure for
$1$-NN searching can be extended to obtain
$O(\log^2 n/\log\log n + k)$ query time, while still retaining
polylogarithmic updates~\cite{Chan12-kNN}. Chan's data structure
essentially maintains a collection of $k$-NN data structures built on
subsets of the sites. A careful analysis shows that some of these
structures can be rebuilt during updates, and that the cost of these
updates is not too large. Queries are then answered by performing
$k_i$-NN queries on several disjoint subsets of sites $S_1,..,S_t$
that together are guaranteed to contain the $k$ nearest sites. However, perhaps because the details of the $1$-NN searching data structure
are already fairly involved, one aspect in the query algorithm is
missing: how to determine the value $k_i$ to query
subset $S_i$ with. While it seems that this issue can be fixed using
randomization~\cite{Chan21Fix}\footnote{The main idea is that the data
  structure as is \emph{can} be used to efficiently report all sites
  within a fixed distance from the query point (reporting all planes
  below a query point in $\R^3$). Combining this with an earlier
  random sampling idea~\cite{Chan00} one can then also answer $k$-NN
  queries.}, our general $k$-NN query technique
(Section~\ref{sec:query_procedure}) allows us to recover deterministic, worst-case
$O(\log^2 n/\log\log n + k)$ query time.

\subparagraph{Organization and Results.} We develop dynamic data
structures for $k$-NN queries in the plane whose query times are of the form
$O(f(n) + k)$, where $f(n)$ is some function of $n$. In particular, we
wish to avoid an $O(k\log n)$ term in the query time. To this end,
we present a general query technique that given $t$ disjoint subsets
of sites $S_1,..,S_t$, each stored in a static data structure that
supports $k'$-NN queries in $O(Q(n)+k')$ time, can report the $k$
nearest neighbors among $\bigcup_{i=1}^t S_i$ in $O(Q(n)t + k)$
time. Our technique, presented in Section~\ref{sec:query_procedure}, is completely
combinatorial, and is applicable to any type of sites. In
Section~\ref{sec:insertion-only}, we then use this technique to obtain
a $k$-NN data structure that supports queries in $O(Q(n)\log n + k)$
time and insertions in $O((P(n)/n)\log n)$ time, where $P(n)$ is the
time required to build the static data structure. This result again
applies to any type of sites. In the specific case of the Euclidean
plane, we obtain a linear space data structure with $O(\log^2 n + k)$
query time and $O(\log^2 n)$ insertion time. At a slight increase of
insertion time, we can also match the
query time of Chan's~\cite{Chan12-kNN} fully dynamic data
structure. For general, constant-complexity, algebraic distance
functions, we obtain the same query and insertion times (albeit the
insertion time holds in expectation). In the case where the sites $S$
are points inside a simple polygon \P with $m$ vertices, we use our
technique to obtain the first \emph{static} $k$-NN data structure that
uses near-linear space, supports efficient (i.e. without the
$O(k\log n)$ term) queries, \emph{and} can be constructed
efficiently. We now do get an $O(k\log m)$ term in the query time, as computing the distance between a pair of points already takes $O(\log m)$ time. Our data structure uses $O(n\log n+m)$ space, can be constructed in $O(n(\log n\log^2 m+\log^3 m) + m)$ time, and supports $O(\log(n+m)\log m + k\log m)$ time
queries. In turn, this then also leads to a data structure supporting efficient, $O(\log^2n\log^2m + \log n\log^3 m)$ time, insertions.

In Section~\ref{sec:Fully_Dynamic} we argue that our general
query algorithm is the final piece of the puzzle for the fully dynamic
case. For the Euclidean plane, this allows us to recover the
deterministic, worst-case $O(\log^2 n/\log\log n + k)$ query time
claimed before~\cite{Chan12-kNN}. Insertions take amortized
$O(\log^{3+\eps}n/\log\log n)$ time, whereas deletions take $O(\log^{5 +\eps}n/\log\log n)$ time. We obtain the same
query and (expected) update times in case of constant degree algebraic distance
functions.

For the geodesic case there is one final hurdle to take. Chan's
algorithm uses a partition-tree based ``slow'' dynamic $k$-NN query data
structure of linear size as one of its subroutines (see Section~\ref{sub:Chan-k-NN}
for details). Liu uses a similar trick (after appropriately
linearizing the distance functions into $\R^c$ for some constant $c$)
in their static $k$-NN data structure~\cite{Liu20}. Unfortunately, this
idea is not applicable in the geodesic setting, as it is unknown if an
appropriate (shallow) simplicial partition exists, and the geodesic
distance function cannot be linearized into a constant dimensional
space (the dimension would need to depend on $m$). Instead, we design
a simple, shallow-cutting based, alternative ``slow'' dynamic $k$-NN
structure, that does extend to the geodesic setting. This way, we end
up with an efficient (i.e. $O(\polylog(n+m))$ expected updates, $O(\log^2n\log^2m + k\log m)$ queries) fully
dynamic geodesic $k$-NN data structure.

\section{Preliminaries}
\label{sec:Preliminaries}

We can easily transform a $k$-nearest neighbors problem in
$\mathbb{R}^2$ to a $k$-lowest functions problem in $\mathbb{R}^3$ by
considering (the graphs of) the distance functions $f_s(x)$ of the
sites $s \in S$. We discuss these problems interchangeably,
furthermore we identify a function with its graph.

\subsection{Shallow cuttings}
\label{sub:Shallow_cuttings}

Let $F$ be a set of bivariate functions. We consider the arrangement of $F$ in
$\mathbb{R}^3$. The \textit{level} of a point $q \in \mathbb{R}^3$ is defined as the number of functions in $F$ that pass strictly below $q$. The \textit{at most }$k$\textit{-level} $L_{\leq k}(F)$ is then the set of points in $\mathbb{R}^3$ that have level at most $k$.

A \emph{$k$-shallow cutting} $\Lambda_k(F)$ of $F$ is a set of
disjoint cells covering $L_{\leq k}(F)$, such that each cell
intersects at most $O(k)$ functions~\cite{m-rph-92}. When $F$ is clear
from the context we may write $\Lambda_k$ rather than
$\Lambda_k(F)$. We are interested only in the case where the cells are
\emph{(pseudo-)prisms}: constant-complexity regions that are bounded
from above by a function in $F$, from the sides by vertical (with respect to
the $z$-direction) surfaces that pass through an intersection curve between two functions in $F$, and unbounded from below. 
For example, if $F$ is a set of planes, we can define the top of each prism to be a triangle. 
This allows us to find the prism containing a query point
$q$ by a point location query in the downward projection of the
cutting. See Figure \ref{fig:shallow_cutting}. The
subset $F_\nabla \subseteq F$ intersecting a prism $\nabla$ is the
\emph{conflict list} of $\nabla$. When, for every subset
$F' \subseteq F$, the lower envelope $L_0(F')$ has linear complexity
(for example, in the case of planes), a shallow cutting of
\emph{size} (the number of cells) $O(n/k)$ can be computed
efficiently~\cite{LiuJournal}. 

In general, let $T(n,k)$ be the time to construct a $k$-shallow cutting of maximum size $S(n,k)$ on $n$ functions, and $Q(n,k)$ be the time to locate the prism containing a query point. We assume these functions are non-decreasing in $n$ and non-increasing in $k$, and that $S(n,k) = \frac{1}{k}f(n)$, for some function $f(n)$.

\begin{figure}
\begin{minipage}{0.3\textwidth}
  \centering
    \includegraphics{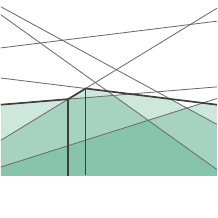}
    \caption{A $2$-shallow cutting of a set of lines $F$ in $\mathbb{R}^2$ consisting of $3$ prisms. The at most $k$-level $L_{\leq k}(F)$ is shown in green for $k=0,1,2$.\newline }
    \label{fig:shallow_cutting}
\end{minipage}
\hfill
\begin{minipage}{0.66\textwidth}
   \centering
    \includegraphics{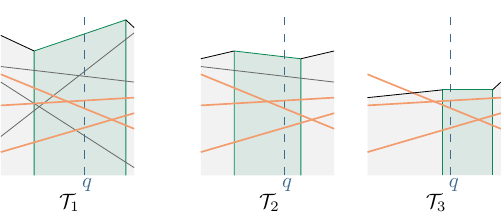}
    \caption{Example of the dynamic 1-NN data structure. Only one shallow cutting ($\Lambda_{k_j}$) is shown for each tower. The orange planes in $\T_1$ and $T_2$ are pruned when building $\Lambda_{k_{j-1}}$, but are not removed from the conflict lists in $\Lambda_{k_j}$. When querying for the $k$-NN, the green prisms in $\Lambda_{k_j}$ of each tower are considered. Note that the three orange planes occur in each of the conflict lists.}
    \label{fig:problem_k-nn}
\end{minipage}
\end{figure}

\subsection{A dynamic nearest neighbor data structure}\label{sub:dynamic-1NN}
We briefly discuss the main ideas of the dynamic nearest neighbor data
structure by Chan \cite{Chan10, Chan19} that was later improved by
Kaplan et al. \cite{Kaplan17}, as this also forms a
key component in our fully dynamic $k$-NN data structures. For a more detailed description we
refer to the original papers. For ease of
exposition, we describe the data structure when $F$ is a set of linear
functions (planes). To make sure the analysis is correct for our definition of $n$ (the current number of points in $S$), we rebuild the data structure from scratch whenever $n$ has doubled or halved. The cost of this is subsumed in the cost of the other operations~\cite{Chan10}.

The data structure consists of $t = O(\log_b n)$ ``towers'' $\tower{1},..,\tower{t}$, for some fixed $b \geq 2$. Each tower $\Ti$ consists of a hierarchy of shallow cuttings, which is built using a process that involves a subset of planes $\Fi \subseteq F$. For $\tower{1}$ we have $\Fx{1} = F$, and a hierarchy of $\ell = \lfloor\log(n/k_0)\rfloor$ shallow cuttings, for a fixed constant $k_0$. For $j = 0,..,\ell$ we have a $k_j$-shallow cutting of a subset of the planes $F_j \subseteq \Fx{1}$, where $k_j = 2^jk_0$. We set $F_{\ell} = \Fx{1}$ and construct these cuttings from $j =\ell$ to $0$. After computing $\Lambda_{k_j}(F_j)$, we find the set $F^{\times}_j$ of ``bad'' planes that intersect more than $c\log n$ prisms in $\Lambda_{k_\ell}(F_\ell),..,\Lambda_{k_j}(F_j)$, i.e. in all cuttings in the hierarchy so far. We \textit{prune} these planes by setting $F_{j-1} = F_j \setminus F^{\times}_j$, and removing all planes in $F^{\times}_j$ from the conflict lists of the prisms in $\Lambda_{k_j}(F_j)$. Note that these bad planes are removed only from the conflict lists of the current cutting, and can still occur in conflict lists of higher level cuttings. In the final $\Lambda_{k_0}(F_0)$ cutting, each conflict list has a constant size of $O(k_0)$. We denote by $\Fbad{1} = F^{\times}_0 \cup ... \cup F^{\times}_\ell$ the set of all bad planes generated during this process. By $\Flive{1}$ we denote the set of planes that have not been pruned during the process, so $\Flive{1} = F^{(1)} \setminus \Fbad{1}$. We then set $\Fx{i+1} = \Fbad{i}$ and recursively build $\tower{i+1}$ on the functions in $\Fx{i+1}$.

In the end, the set $F$ is partitioned into sets $\Flive{1},..,\Flive{t}$. When insertions and deletions take place, planes can move from a set $\Flive{i}$ to some $\Flive{i'}$, but the property that these sets form a partition of $F$ will be preserved. Kaplan et al. prove the following lemma on the size of $\Flive{1}$ after the preprocessing phase:

\begin{lemma}[Lemma 7.1 of \cite{Kaplan17}]\label{lem:Kaplan_pruning}
For any $\zeta \in (0, 1)$ there exists a sufficiently large (but constant) choice
of $c$, such that $|\Flive{1}| \geq (1-\zeta)n$ after building $\tower{1}$.
\end{lemma}

When $\zeta = 1/b$, we get $O(\log_b n)$ towers, for some fixed $b \geq 2$,
as desired. According to Kaplan \etal~\cite{Kaplan17}, this is achieved by choosing $c \geq \frac{\gamma}{\zeta} = b\gamma = O(b)$, for some constant $\gamma$. Thus a plane occurs $O(b\log n)$ times in a
tower. Here, we first consider only the case where $b = 2$.

To build a single tower, naively we would need to compute $O(\log n)$ shallow cuttings, each of which takes $O(n \log n)$ time. By using information of previously computed cuttings, Chan~\cite{Chan19} recently achieved an overall construction time of $O(n \log n)$. The preprocessing time of the entire data structure thus adheres to the recurrence $P(n) \leq P(n/2) + O(n \log n)$. This solves to $P(n) = O(n \log n)$.

\subparagraph{Insertions.}
To insert a plane $f$ into $F$, we insert it into $\Fx{1}$. When we insert a function $f$ into $F^{(i)}$ we assign it to $\Fbad{i}$, and thus recursively insert it
in $F^{(i+1)}$. When $|\Fbad{i}|$ reaches $3/4 \cdot |\Fi|$ we rebuild the towers $\Ti,..,\tower{t}$. The first tower, $\Ti$, is built on the planes $\{f\} \cup \Flive{i} \cup ... \cup \Flive{t}$, and the following towers are again built recursively on the new sets $\Fbad{i'-1}$. Only after $\Omega(|\Fi|)$ insertions can such a rebuild occur. The insertion time is thus given by the recurrence $I(n) \leq I(3n/4) + O(P(n)/n)$, where $P(n)$ is the time to build the data structure on a set of $n$ planes. Using $P(n) = O(n\log n)$, this results in an amortized insertion time of $O(\log^2n)$.

\subparagraph{Deletions.}
Deletions are not performed explicitly on the conflict lists. Instead,
for each prism $\nabla$ we keep track of the number of planes in
$F_{\nabla}$ that have been deleted so far, denoted by
$d_{\nabla}$. When deleting a plane $f$, we increase $d_{\nabla}$ for
all prisms with $f \in F_{\nabla}$, and remove $f$ from the set
$\Flive{i}$ that includes $f$. When too many planes in a conflict
list have been deleted, we \textit{purge} the prism. In particular,
we purge a prism $\nabla$ in a $k_j$-shallow cutting when $d_{\nabla}
\geq k_j/2 =2^{j-1}k_0$. When a prism in $\Ti$ is purged, we mark it
as such, and we reinsert all planes $f'\in F_{\nabla} \cap
\Flive{i}$. These planes are effectively moved from $\Flive{i}$ to
some other $\Flive{i'}$. Note that we only reinsert planes that have
not been deleted so far. This scheme
ensures a prism is only purged after at least $k_j/2$ deletions, and
this causes at most $|F_\nabla| = O(k_j)$ reinsertions. Thus, each increment of $d_\nabla$ causes amortized $O(1)$ reinsertions. This gives an amortized deletion time of $O(\log^4n)$.

\subparagraph{Nearest neighbor queries.}
When answering a nearest neighbor query for a query point $q$, we simply find the prism containing $q$ in the lowest ($j = 0$) cutting of each $\Ti$ by a point location query in $O(\log n)$ time. We then go through each conflict list (of constant size) to find the plane that is lowest at $q$. If the plane we find for $\Ti$ is not in $\Flive{i}$, we ignore the result. When a prism has been purged, we simply skip it.  Finally, we return the plane that is lowest among the $O(\log n)$ planes that are found. As we perform $O(\log n)$ point location queries, the query time is $O(\log^2n)$.

\subparagraph{$k$-nearest neighbors queries.}

Answering $k$-nearest neighbors queries using this dynamic 1-NN data structure is straightforward. For each tower we consider the prism containing $q$ of the shallow cutting at level $j_k := \lceil \log(Ck/n)\rceil$, for some large enough constant $C$. The size of the conflict list of each of these prisms is $O(k)$, thus we can find the $k$-lowest \textit{live} planes in each conflict list in $O(k)$ time. Chan~\cite{Chan10} proves that it is indeed sufficient to consider only planes in these conflict lists. This query algorithm has a running time of $O(\log^2n +k\log n)$.

However, even if we were to know the exact number of the $k$-nearest
neighbors that occurs in each tower, we would not be able to support
$k$-NN queries in $O(\log^2 n + k)$ time. When a plane is pruned during the preprocessing, or when a prism is purged, the plane is only removed from the conflict lists of the current shallow cutting. It can thus still occur in other shallow cuttings in the hierarchy. This means that we can encounter the same plane multiple times when querying each tower for the $k$-lowest planes. See Figure \ref{fig:problem_k-nn} for an illustration. As there are $O(\log n)$ towers, this yields an $O(k\log n)$ term in the query time.

\subparagraph{General distance functions.}
Kaplan \etal~\cite{Kaplan17} showed how to adapt Chan's data structure to support more general shallow cutting algorithms. The main differences between their data structure and the one from Chan is that planes are only pruned when they appear in $(S(n,1)/n)\cdot c\log n$ conflict lists. Kaplan \etal essentially prove the following lemma.

\begin{lemma}[Kaplan \etal~\cite{Kaplan17}]\label{lem:kaplan}
    Given an algorithm that constructs a $k$-shallow cutting of size $S(n,k)$ on $n$ functions in $T(n,k)$ time, such that the prism containing a query point can be located in $Q(n,k)$ time, we can construct a data structure of size $O(S(n,1)\log n)$ that dynamically maintains a set of at most $n$ functions $F$. Reporting the lowest function at a query point $q$ takes $O(Q(n,1)\log n)$ time, inserting a function in $F$ takes $O((T(n,1)/n)\log^2 n)$ amortized time, and deleting a function from $F$ takes $O((T(n,1)S(n,1)/n^2) \log^4 n)$ amortized time.
\end{lemma}
  
From now on we consider the general variant, where the data structure consists of $O(\log_b n)$ towers. The following lemma describes the properties of the $1$-NN data structure we need to construct our general fully dynamic $k$-nearest neighbors data structure.

\begin{lemma}\label{lem:properties}
    Let $b \geq 2$ be any fixed value and $S(n,k)$ be the maximum size of a $k$-shallow cutting. There is a dynamic nearest neighbor data structure that has the following properties.
   \begin{enumerate}
        \item The data structure consists of $O(\log_b n)$ towers.
        \item A function occurs $O(b\log n \cdot S(n,1)/n)$ times in a conflict list in a single tower.
        \item The insertion time is $O(b \log_b n \cdot P(n)/n)$, where $P(n)$ is the preprocessing time.
        \item A deletion causes amortized $O(b \log_b n\log n \cdot S(n,1)/n)$ reinsertions.
        \item To find the $k$-NN of a query point $q$, it is sufficient to consider $O(\log_b n)$ prisms, namely for each tower the prism containing $q$ of the shallow cutting at level $j_k := \lceil \log(Ck/n)\rceil$, for some large enough constant $C$.
    \end{enumerate}
\end{lemma}

\section{Querying multiple $k$-NN data structures simultaneously}
\label{sec:query_procedure}

In this section we introduce a method to find the $k$-nearest neighbors of a query point $q$ spread over $t$ (disjoint) $k'$-NN data structures storing a set of sites $S$ simultaneously. Suppose the query time of such a $k'$-NN data structure is $O(Q(n) + k')$, for a non-decreasing function $Q$. Naively,
querying each data structure for the $k$ closest sites would take
$O(Q(n)t + tk)$ time. Our method allows us to find the $k$-NN over
all these data structures in $O(Q(n)t + k)$ time instead, thus reducing the $O(tk)$ term to $O(k)$. More formally, we prove the following result.

\begin{theorem}
  \label{thm:tsets-knn}
  Let $S_1,..,S_t$ be disjoint sets of point sites of sizes
  $n_1,..,n_t$, each stored in a data structure that supports $k'$-NN
  queries in $O(Q(n_i)+k'T)$ time, where $T$ is the time for evaluating $d(p,q)$. There is a $k$-NN data structure on
  $\bigcup_i S_i$ that supports queries in $O(Q(n)t + kT)$ time. The
  data structure uses $O(\sum_i C(n_i))$ space, where $C(n_i)$ is the
  space required by the $k$-NN structure on $S_i$.
\end{theorem}

\begin{proof}
We first describe the algorithm to query all $t$ data structures simultaneously, and then analyse its running time.

\subparagraph{Query algorithm.}
\label{sub:Query_algorithm}
We use the heap selection algorithm of
Frederickson~\cite{Frederickson93} to answer $k$-NN
queries efficiently. This algorithm finds the $k$
smallest elements of a binary min-heap of size $N \gg k$ in $O(k)$
time. Running this algorithm on a heap that contains all sites $s \in S$ exactly once, with the distance $d(s,q)$ as key for each site, would return the $k$-nearest neighbors. However, building this entire heap takes linear time. To overcome this issue, we do not
construct the entire heap we query before starting the algorithm. Instead,
the heap is expanded during the query when necessary. See
Figure~\ref{fig:expand_heap} for an example.

To know when to expand the heap, we require that the Frederickson algorithm only visits a node \emph{after} the parent of that node has been visited. This implies that once a leaf of the heap is visited by the algorithm, we can expand the heap further below that leaf without hindering the algorithm.

  \begin{figure}
      \centering
      \includegraphics{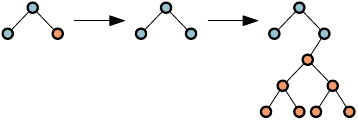}
      \caption{Example of expansion. Blue elements have been visited by the algorithm, orange elements have not. The expansion (building the
        next subheap) occurs when the last element is visited.}
      \label{fig:expand_heap}
  \end{figure}

  Next, we describe the heap $H$, on which we call the heap selection algorithm, in more detail. As stated before, $H$ contains
  all sites $s \in S$ exactly once, with the distance $d(s,q)$ as key
  for each site.
  Let $S_1,..,S_t$ be the partition of $S$ into $t$ disjoint sets, where $S_j$ is the set of sites stored in the $j$-th $k'$-NN data structure. For
  each set of sites $S_j$, $j \in \{1,..,t\}$, we define a heap
  $H(S_j)$ containing all sites in $S_j$. We then ``connect'' these
  $t$ heaps by building a dummy heap $H_0$ of size $O(t)$ that has the
  roots of all $H(S_j)$ as leaves. We set the keys of the elements of
  $H_0$ to $-\infty$. Let $H$ be the complete data structure (heap)
  that we obtain this way, see Figure \ref{fig:frederickson}. It
  follows that we can now compute the $k$ sites closest to $q$ by
  finding the $|H_0| + k$ smallest elements in the resulting heap $H$
  and reporting only the non-dummy sites.

What remains is how to (incrementally) build the heaps $H(S_j)$ while
running the heap selection algorithm. Each such heap consists of a
hierarchy of \textit{subheaps} $H_1(S_j),..,H_{O(\log n)}(S_j)$, such that every element of $S_j$ appears in exactly one $H_i(S_j)$. Moreover, since the sets $S_1,..,S_j$ are pairwise disjoint, this holds for any $s \in S$, i.e. $s$ appears in exactly one $H_i(S_j)$. The \textit{level 1} heaps, $H_1(S_j)$, consist of the $k_1 = Q(n)$ sites in $S_j$ closest to $q$, which we
find by querying the static data structure of $S_j$. The subheap $H_i(S_j)$ at level
$i > 1$ is built only after the last element $e$ of $H_{i-1}(S_j)$ is visited by the heap selection algorithm.  We then add a pointer from
$e$ to the root of $H_i(S_j)$, such that the root of $H_i(S_j)$
becomes a child of $e$, as in Figure \ref{fig:expand_heap}.

To construct a subheap $H_i(S_j)$ at level $i > 1$, we query the
static data structure of $S_j$ using $k_i = k_1 2^{i-1}$. The new
subheap is built using all sites returned by the query that have not
been encountered earlier. It follows that all elements of $H_i(S_j)$
are larger than any of the elements in
$H_1(S_j),..,H_{i-1}(S_j)$. Thus, the heap property is preserved.

\begin{figure}
    \centering
    \includegraphics{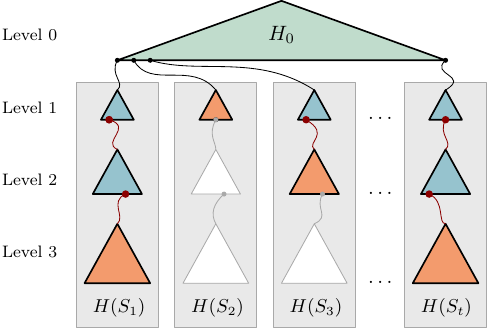}
    \caption{The heap that we construct for the $k$-nearest neighbors
      query. The subheaps of which all elements have been visited are indicated in \textit{blue}. The subheaps that have been built,
      but for which not all elements have been visited, are
      indicated in \textit{orange}. The \textit{white} subheaps have not been built so far,
      because not all elements of their predecessor have been visited.}
    \label{fig:frederickson}
\end{figure}

\subparagraph{Analysis of the query time.}
\label{sub:Query_analysis}

As stated before, finding the $k$-smallest non-dummy elements of $H$ takes $O(k + |H_0|)$ time~\cite{Frederickson93}. 
We now analyse the time used to construct $H$.

First, the level 0 and level 1 heaps are built. Building $H_0$ takes only $O(t)$ time. To build the level 1 heaps, we query each of the substructures using $k_1 = Q(n)/T$, where $T$ denotes the time to evaluate the distance function. In total these queries take $O((Q(n) + k_1T)t) = O(Q(n)t)$ time. 
Retrieving the next $k_i$ elements to build $H_i(S_j)$ for $i > 1$ requires a single query, and thus takes $O(Q(n)+k_iT)$ time. 
To bound the time used to build all heaps at level greater than 1, we first prove the following lemma.

\begin{lemma}\label{lem:sizeHi}
    The size of a subheap $H_i(S_j)$, $j \in \{1,..,t\}$, at level $i > 1$ is exactly $k_1 2^{i-2}$.
\end{lemma}

\begin{proof}
    To create $H_i(S_j)$, we query the static data structure of $S_j$ to find the $k_1 2^{i-1}$ sites closest to $q$. Of these sites, only the ones that have not been included in any of the lower level subheaps are included in $H_i(S_j)$. The sites previously encountered are exactly the $k_1 2^{i-2}$ sites returned in the previous query. It follows that $|H_i(S_j)| = k_1 (2^{i-1} - 2^{i-2})= k_1 2^{i-2}$. 
\end{proof}

Building $H_i(S_j)$ takes $O(Q(n)+k_iT)$ time. To pay
for this, we charge $O(T)$ to each element of $H_{i-1}(S_j)$. Because
we choose $k_1 = Q(n)/T$, Lemma~\ref{lem:sizeHi} implies that
$|H_{i-1}(S_j)| = \Omega(Q(n)/T)$,  and that
$k_i = k_1 2^{i-1} = 2^2 k_1 2^{i-3} = O(|H_{i-1}(S_j)|)$. 
Note that the heap $H_i(S_j)$, $i > 1$, is only built when the final
element of $H_{i-1}(S_j)$ is visited.  Thus, we only charge elements
of the heaps of which all elements have been visited (shown blue in
Figure \ref{fig:frederickson}). In total, $O(k)$ elements (not in
$H_0$) are visited, so the total size of these subheaps is
$O(k)$. From this, and the fact that all subheaps are disjoint, it
follows that we charge $O(T)$ to only $O(k)$ sites. The total running
time thus becomes $O(t)  + O(Q(n)t) + O(kT) = O(Q(n)t + kT)$.
This completes the proof of Theorem~\ref{thm:tsets-knn}.
\end{proof}

\section{An insertion-only data structure}
\label{sec:insertion-only}

We describe a method that transforms a static $k$-NN data structure
with query time $O(Q(n) + k)$ into an insertion-only $k$-NN data
structure with query time $O(Q(n)\log n + k)$. Insertions take
$O((P(n)/n)\log n)$ time, where $P(n)$ is the preprocessing time of
the static data structure, and $C(n)$ is its space usage. We assume
$Q(n)$, $P(n)$, and $C(n)$ are non-decreasing.

To support insertions, we use the logarithmic
method~\cite{Overmars83}. We partition the sites into $O(\log n)$
groups $S_1,..,S_{O(\log n)}$ with $|S_i| = 2^i$ for
$i \in \{1,..,O(\log n)\}$. To insert a site $s$, a new group
containing only $s$ is created. When there are two groups of size
$2^i$, these are removed and a new group of size $2^{i+1}$ is
created. For each group we store the sites in the static $k$-NN data
structure. This results in an amortized insertion time of
$O((P(n)/n)\log n)$. This bound can also be made
worst-case~\cite{Overmars83}. The main remaining issue is then how to support queries in $O(Q(n)\log n + k)$ time, thus avoiding an
$O(k\log n)$ term in the query time. 
Applying Theorem \ref{thm:tsets-knn} directly solves this problem, and we thus obtain the following result.

\begin{theorem}\label{thm:insertionsonly-ds}
  Let $S$ be a set of $n$ point sites, and let $\D$ be a static $k$-NN
  data structure of size $O(C(n))$, that can be built in $O(P(n))$
  time, and answer queries in $O(Q(n) + k)$ time. There is an
  insertion-only $k$-NN data structure on $S$ of size $O(C(n))$ that
  supports queries in $O(Q(n)\log n + k)$ time. Inserting a new site
  in $S$ takes $O((P(n)/n)\log n)$ time.
\end{theorem}

\subsection{Points in the plane}
In the Euclidean metric, $k$-nearest neighbors queries in the plane can be answered in $O(\log n + k)$ time, using $O(n)$ space and $O(n\log n)$ preprocessing time \cite{Afshani09,Chan16}. Hence:

\begin{corollary}
    There is an insertion-only data structure of size $O(n)$ that stores a set of $n$ sites in $\mathbb{R}^2$, allows for $k$-NNs queries in $O(\log^2 n + k)$ time, and insertions in $O(\log^2 n)$ time.
  \end{corollary}

If we increase the size of each group in the logarithmic method to
  $b^i$, with $b = \log^\eps n$ and $\eps > 0$, we get only
  $O(\log_b n)$ groups instead of $O(\log n)$. This reduces the query
  time to $O(\log^2 n/ \log \log n + k)$, matching the fully dynamic
  data structure. However, this also increases the insertion time to
  $O(\log^{2 + \eps}n/\log \log n)$. 
  For general constant-complexity
  distance functions, we achieve the same query time using Liu's data
  structure~\cite{LiuJournal}. The space usage is $O(n\log \log
  n)$ and the expected insertion time is $O(\log^2 n)$.

\subsection{Points in a simple polygon}\label{sec:geodesicknn_insertion_only}
In
the geodesic $k$-nearest neighbors problem, $S$ is a set of sites
inside a simple polygon $\mathcal{P}$ with $m$ vertices. For any two points $p$
and $q$ the distance $d(p,q)$ is defined as the length of the shortest
path $\pi(p,q)$ between $p$ and $q$ fully contained within $\mathcal{P}$. The input polygon
$\mathcal{P}$ can be preprocessed in $O(m)$ time so that the geodesic distance $d(p,q)$ between any two
points $p,q \in \mathcal{P}$ can be computed in $O(\log m)$
time~\cite{Guibas89}.

To apply Theorem~\ref{thm:insertionsonly-ds}, we need a
static data structure for geodesic $k$-NN queries. We can build such a
data structure by combining the approach of Chan~\cite{Chan00} and
Agarwal \etal~\cite{Staals18}. The data structure consists of a
hierarchy of lower envelopes of random samples
$R_0 \subset R_1 \subset .. \subset R_{\log n}$. For each sample, we
store a (topological) vertical decomposition of the downward
projection of the lower envelope, and the conflict lists of the corresponding
pseudo-prisms. The downward projection is a geodesic Voronoi diagram, which can be preprocessed for
efficient point location queries using the method of Oh and
Ahn~\cite{Oh20}. From a Clarkson and Shor style sampling argument, it
follows that total expected size of all conflict lists of one sample
is $O(n)$. Thus, the space of the resulting data structure is
$O(n\log n)$ in expectation. We can then find a prism in one
of the vertical decompositions that contains the query point and whose
conflict list has size $O(k)$ in $O(\log(n+m) + k\log m)$
time~\cite{Chan00}. This allows us to answer $k$-NN queries in the
same time. The crux in this approach is in how to compute the conflict
lists. We can naively compute these in $O(mn)$ time by
explicitly constructing the geodesic distance function for each site,
and intersecting it with each of the $O(n)$ pseudo-prisms. It is unclear how
to improve on this bound.

\begin{theorem}\label{thm:static-geodesic-knn}
    Let $S$ be a set of $n$ sites in a simple polygon \P with $m$ vertices. In $O(n(\log n \log^2 m + \log^3 m))$ time we can build a data structure of size $O(n\log n\log m)$, excluding the size of the polygon, that can answer $k$-NN queries with respect to $S$ in $O(\log(n+m)\log m + k\log m)$ time.
\end{theorem}

\begin{proof}
To circumvent the issue we describe above, we partition $\mathcal{P}$ into two subpolygons $\mathcal{P}_r$ and $\mathcal{P}_\ell$ of roughly the same size by a diagonal $d$. Let $S_r$ and $S_\ell$ denote the sites in $\mathcal{P}_r$ and $\mathcal{P}_\ell$, respectively. We then build a data structure that can find the $k$-nearest neighbors of any query point $q \in \P_r$ among the sites in $S_\ell$ efficiently. To allow us to answer general queries, i.e. to find the $k$-nearest neighbors for \emph{any} $q \in \P$ among \emph{all} sites in $S$, we build such a data structure on both $\P_r$ and $\P_\ell$, and then recursively partition the polygon further. This results in a decomposition of the polygon of $O(\log m)$ levels.

\begin{figure}
    \centering
    \includegraphics{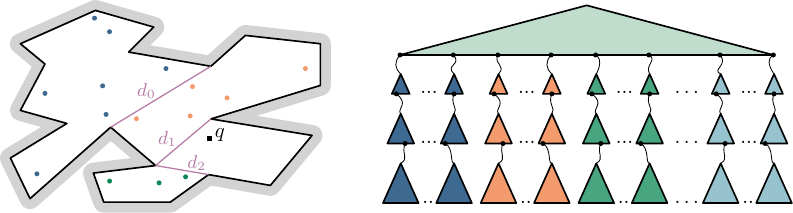}
    \caption{A partial decomposition of $\mathcal{P}$ and the corresponding heap
      used in a $k$-NN query for $q$.
    }
    \label{fig:geodesic-heap}
\end{figure}

To answer $k$-NN queries for a point in $\P_r$ among the sites in $S_\ell$, we consider the Voronoi diagram of sites in $S_\ell$ (resp. $S_r$) restricted to
$\mathcal{P}_r$. The bisector between any two sites in $S_\ell$ restricted to $\P_r$ is $x$-monotone with respect to the vertical diagonal $d$. Agarwal \etal use this fact to construct an
efficient data structure (Theorem 22 of \cite{Staals18}) for our
problem. This is essentially the data structure that was described at
the start of this section. However, because the bisectors are $x$-monotone, we can efficiently compute the conflict
lists by only considering the functions intersecting the corners of
each prism. Because of the monotonicity, a bisector cannot enter and exit again through the same edge of a prism. Therefore, the conflict list of a prism is simply the union of the conflict lists of its corners (Lemma 15 of~\cite{Staals18}). For each function $f$, we compute the vertices it conflicts with by first finding a vertex $v$ on $d$ that conflicts with $f$, and then performing a breadth first search to find \emph{all} vertices (corners) it conflicts with.
Thus, we can build the data structure in $O(n(\log n \log m + \log^2 m))$
time (see \cite{Staals18} for details). The data structure requires $O(n\log n)$ space, excluding the size of
the polygon, and it allows us to find the $k$-nearest neighbors among $S_\ell$ for a point $q \in \mathcal{P}_r$ in $O((\log n + k)\log m)$ expected time~\cite{Staals18}. 

To answer a $k$-NN query in $\P$, we first check whether the query point $q$ lies in $\P_r$ or $\P_\ell$. When $q$ lies in $\P_r$ (resp. $\P_\ell$), we query the data structure on $S_\ell$ to find the $k$-nearest neighbors among this set. To find the $k$-nearest neighbors among $S_r$, we recursively query $\P_r$. At each level of the decomposition, we thus consider a data structure on a single set of sites ($S_\ell$ or $S_r$) (see Figure~\ref{fig:geodesic-heap}). In the next paragraph, we describe how to improve the query time of the static data structure, which can answer $k$-NN queries for a query point in $\P_r$ among sites in $S_\ell$, to $O(\log(n+m) + k\log m)$ time. It follows that, using our technique from Section~\ref{sec:query_procedure}, we can find the $k$-NN in $S$, which are spread over $O(\log m)$ of these data structures, in $O(\log(n+m)\log m + k\log m)$ time.

\subparagraph{Improving the query time.} The query time of the static data structure is determined by two factors: a point location query in the topological vertical decomposition of a geodesic Voronoi diagram, this takes $O(\log n \log m)$ time, and $O(k)$ distance queries that take $O(\log m)$ time.
We can improve the point location time to $O(\log(n+m))$ by
incorporating the idea of Oh and Ahn~\cite{Oh20} to approximate a
geodesic Voronoi diagram by a polygonal subdivision.
Given the exact location of the degree-1 and degree-3 vertices of the Voronoi diagram, they approximate each common boundary of two Voronoi regions by connecting the two endpoints using at most three line segments. This allows them to find the (not-approximated) Voronoi region that contains a query point $q$ in $O(\log(n + m))$ time, using $O(n \log(n + m))$ preprocessing time. In our case, we want to find the (pseudo-)prism that contains the query point. This corresponds to finding the (pseudo-)trapezoid, which is the downward projection of the prism, of the (topological) vertical decomposition of a geodesic Voronoi diagram that contains $q$. Thus, we need to slightly adapt their approach to not only find the Voronoi region of the query point $q$, but also the exact trapezoid containing $q$.

Instead of approximating the bisector connecting two Voronoi vertices, we approximate each part of the bisector between two vertices of the vertical decomposition separately, see Figure~\ref{fig:approximate_voronoi}. To approximate a bisector of sites $s_1$ and $s_2$ between vertices $u,v$, Oh and Ahn first find two points $t_1,t_2$ such that the geodesic convex hull of $t_1,t_2,u,v$ is contained in the Voronoi regions of $s_1$ and $s_2$, and the boundary of this convex hull consists of at most four maximal concave polygonal curves. They then approximate the bisector by either the line segment $uv$, when this is contained in the convex hull, or a polygonal curve consisting of extensions of the edges incident to $u$ and $v$ of $\pi(u,v)$, and a line segment tangent to $\pi(t_1,v)$ and $\pi(t_2,u)$ that connects them.

We now consider the vertical decomposition of our Voronoi diagram. We want to approximate a bisector of sites $s_1$ and $s_2$ between two vertices $u,v$ of the vertical decomposition. Without loss of generality, let $v$ be the vertex furthest from $d$. Oh and Ahn choose $t_i$ as the junction of $\pi(s_i,u)$ and $\pi(s_i,v)$. When choosing the $t_i$'s like this, it could be that either of the $t_i$'s is not contained in the trapezoid whose boundary we are approximating. In this case, we instead choose $t_i$ as the intersection between $\pi(s_i,v)$ and the left line segment defining the trapezoid, see Figure \ref{fig:approximate-bisector}. As $\pi(s_i,v)$ is contained in the Voronoi region of $s_i$, this intersection point indeed exists. Note that this also ensures the convex hull is still bound by at most four maximal concave curves. Thus, we can use we can use the algorithm of Oh and Ahn~\cite{Oh20} to approximate the bisector within the convex hull.
\begin{figure}
    \begin{minipage}{0.48\textwidth}
    \centering
    \includegraphics{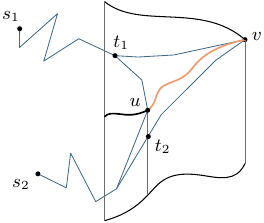}
    \vspace{0.6cm}
    \caption{The four convex chains between $u,t_1,v,t_2$ are used to approximate the orange bisector between $u$ and $v$.}
    \label{fig:approximate-bisector}
    \end{minipage}
    \hfill
    \begin{minipage}{0.48\textwidth}
      \centering
      \includegraphics{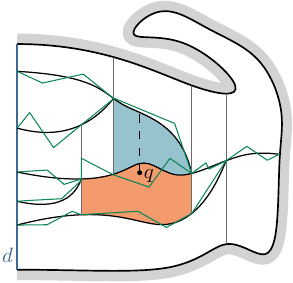}
      \caption{Approximation (in green) of the Voronoi diagram in $\mathcal{P}_r$. To find the trapezoid containing $q$, we consider both colored trapezoids.}
      \label{fig:approximate_voronoi}
    \end{minipage}
  \end{figure}
  
To find the trapezoid containing a query point $q$, we use the query algorithm of Oh and Ahn~\cite{Oh20}. The query is performed as follows: first, we find the approximated trapezoid containing $q$ by shooting a ray upwards from $q$, and
determining what segment in the approximated Voronoi diagram (or the
polygon boundary) is hit. Let $s$ be the site whose Voronoi region we find. Suppose this is not the real trapezoid containing $q$, then $q$ lies in a region bounded by part of a bisector of $s$ and some other site $t$, and the approximation of that bisector, see Figure \ref{fig:approximate_voronoi}. To find this approximated bisector (and thus the trapezoid in the Voronoi region of $t$ containing $q$), we shoot a ray from $q$ in direction opposite to the edge of $\pi(s,q)$ incident to $q$. This does not intersect the real bisector, and thus intersects the approximated bisector of $s$ and $t$. Finally, we compare the distance from $q$ to the two sites to find the trapezoid of the vertical decomposition containing $q$.
\end{proof}

Applying Theorem \ref{thm:tsets-knn} to the data structure of Theorem \ref{thm:static-geodesic-knn} proves the following corollary.

\begin{corollary}
Let $\mathcal{P}$ be a simple polygon with $m$ vertices. There is an insertion-only data structure of size $O(n \log n \log m + m)$ that stores a set of $n$ point sites in $\mathcal{P}$, allows for geodesic $k$-NN queries in $O(\log(n+m)\log n \log m + k \log m)$ expected time, and inserting a site in $O(\log^2 n \log^2 m + \log n \log^3 m)$ time.
\end{corollary} 

\section{A fully dynamic data structure}
\label{sec:Fully_Dynamic}
In this section, we develop a $k$-NN data structure that supports both insertions and deletions, building on the results of
Chan~\cite{Chan12-kNN}. The data structure we propose is an extension of the nearest neighbor data structure of Lemma~\ref{lem:kaplan}, and achieves the following result.\vspace{1cm}

\begin{restatable}{theorem}{thmkNNgeneral}
  \label{thm:k-NN_general}
    Given an algorithm that constructs a $k$-shallow cutting of size $S(n,k)$ on $n$ functions in $T(n,k)$ time, such that the prism containing a query point can be located in $Q(n,k)$ time, we can construct a data structure that dynamically maintains a set of at most $n$ functions $F$ for $k$-NN queries. The size of the data structure is $O(S(n,1)^2/n \cdot \log n\log\log n)$, and operations are performed in time:
    
    \vspace{0.5em} \hspace*{-1.5em}
    \begin{tabular}{ll}
        Query: & $O(Q(n,1)\log n/ \log\log n + k)$ \\
        Insertion: & $O((S(n,1)T(n,1)/n^2)\log^{2+\eps}n)^*$\\
        Deletion: & $O((T(n,1)/\log \log n + T(n,n/\log n)) \cdot S(n,1)^2\log^{4+\eps}n /n^3)^*$\\
        & \hspace{10.3cm}\small{$^*$amortized}
    \end{tabular}
\end{restatable}

In Section~\ref{sub:Chan-k-NN}, we first describe the $k$-NN data structure for planes by Chan~\cite{Chan12-kNN}, and then fill in the part missing from
Chan's query algorithm. This data structure forms the basis for our general $k$-NN data structure. In Section~\ref{sub:Bootstrapping}, we then discuss a simple deletion-only
$k$-NN structure. This deletion-only structure allows us to
adapt Chan's $k$-NN data structure
to more general distance functions like the geodesic distance. In Section~\ref{sub:improved-k-NN}, we analyze the running times of our general $k$-NN data structure with respect to the Euclidan distance, as this is somewhat easier to follow. In Section~\ref{sub:general_dynamic_k-nn}, we prove the general running times with respect to any vertical shallow cuttings as stated in Theorem~\ref{thm:k-NN_general}. Finally, in Section~\ref{sec:fully_dynamic_applications}, we apply Theorem~\ref{thm:k-NN_general} to the case of general constant-complexity distance functions and the geodesic distance.

\subsection{A dynamic $k$-NN data structure for planes}
\label{sub:Chan-k-NN}

Chan \cite{Chan12-kNN} describes how to adjust his original 1-NN data
structure to efficiently perform $k$-NN queries. We denote this (adjusted) data
structure by $\D$. There are two main changes in this data structure:
the conflict lists are stored in $k$-NN data structures, and the
number of towers is reduced by using $b = \log^\eps n$. Only the \textit{live} planes of the conflict list of each prism $\nabla$ are stored in a data
structure $\D_0$ that uses linear space, can perform $k'$-NN (or $k'$-lowest planes) queries in $O(Q_0(|F_\nabla|) +
k')$ time, and deletions in $D_0(|F_\nabla|)$ time. A different data structure is used to store small and large conflict lists. After building a tower $\Ti$, each data structure $\D_0$ of a prism $\nabla$ is built on $F_\nabla \cap \Flive{i}$. As the $\D_0$ data structures use linear space, the space usage of the entire data structure is $O(n\log n)$.

Insertions are performed as in the original data structure, but the deletion of a plane $h$ requires more effort than in the original. 
In addition to increasing $d_{\nabla}$ for each prism containing $h$, $h$ is explicitly removed from the $\D_0$ data structures. Note that $\Ti$ for which $h \in \Flive{i}$ is the only tower whose $\D_0$ data structures contain $h$. When a prism in tower $\Ti$ is purged, we also delete its planes from the other $\D_0$ data structures in $\Ti$ to retain this property. This gives an amortized expected update time of $U(n) = O(\log^{6+\eps} n)$~\cite{Chan12-kNN}. The improvement of Kaplan \etal~\cite{Kaplan17} reduces this to $O(\log^{5+\eps}n)$, and the improvement of Chan~\cite{Chan19} makes it deterministic. In Section \ref{sub:improved-k-NN} we discuss this update time in more detail w.r.t. our adaption of the data structure.
It follows from Lemma~\ref{lem:properties} and the above modifications that:
\begin{lemma}[Chan~\cite{Chan12-kNN}]
Let $q$ be a query point. In $O(t\log n)$ time, we can find $t=O(\log_b n)$ prisms $\nabla_1,..,\nabla_t$, such that: (i) all prisms contain $q$, (ii) the conflict list of each prism has size $O(k)$, (iii) the conflict lists are pairwise disjoint and stored in a $\D_0$ data structure, and (iv) the $k$ sites in $S$ closest to $q$ appear in the union of the conflict lists of those prims.
\end{lemma}

So, to answer $k$-NN queries we can use a $k_i$-NN query on each $\D_0$ data structure of the prisms $\nabla_1,..,\nabla_t$, where $k_i$ is the number of sites from the $k$-nearest neighbours of $q$ that appear in the conflict list of $\nabla_i$. This takes $O(\sum_{i=1}^t Q_0(k) + k_i) = O(\sum_{i=1}^t Q_0(k) + k)$ time. However, it is unclear how to compute those $k_i$ values. Fortunately, we can use Theorem \ref{thm:tsets-knn} to find the $k$-nearest neighbors over all of the substructures in $O(Q_0(k)\log_b n + k)$ time. Plugging in the appropriate query time $Q_0(k)$ (see Chan~\cite{Chan12-kNN} and Section~\ref{sub:improved-k-NN}), this achieves a total query time of $O(\log^2 n/\log\log n + k)$ time as claimed.

\subsection{A simple deletion-only data structure}
\label{sub:Bootstrapping}

Let $H$ be a set of $n$ planes, and let $r \in \mathbb{N}$ be a
parameter. We develop a data structure that
supports deletions, and reporting the $t$ lowest planes above a query point $q \in \R^2$. For our application, we are mostly interested in the deletion time of the data structure, and less in the query time. By picking $r$ to be somewhat small, we can make deletions efficient at the cost of making the query time fairly terrible.

\begin{lemma}\label{lem:bootstrapping_ds}
    For any fixed $r$, we can construct a deletion-only data structure $\D_0$ of size $O(n\log r)$, or $O(n)$ when $n \geq r^{1/\eps}$, in $O(n\log n)$ time that stores a set of $n$ planes, allows for $t$-lowest planes queries in $O(\log r + n/r + t)$ time and deletions in $O(r\log n)$ amortized time. 
\end{lemma}

\begin{proof}
Our entire data structure consists of
just $\ell=O(\log r)$ $k_i$-shallow-cuttings $\Lambda_{k_0},..,\Lambda_{k_\ell}$ of the planes, for values $k_i = 2^i \lceil n/r \rceil$, for
$i=0,..,\ell$. Hence, this uses $O(n\log r)$ space in total. We can compute the shallow cuttings along with their conflict lists in $O(n\log n)$ time~\cite{Chan16}. Note that even when $r > n$, we only need to build $O(\log n)$ shallow cuttings.

\subparagraph{Deletions.}
If we delete a plane, we remove it from all conflict
lists in all cuttings. Since cutting $\Lambda_{k_i}$ has size $O(r/2^i)$,
each plane occurs at most $O(r/2^i)$ times. Hence, the total time to go
through all of these prisms is $\sum_{i=0}^{\log r} r/2^i = O(r)$ time. When more than half of the planes from any conflict list are removed, we rebuild the entire data structure. Because the smallest conflict list contains at least $n/r$ planes, at least $n/2r$ deletions take place before a global rebuild. We charge the $O(n\log n)$ cost of rebuilding to these
planes, so we charge $O(r\log n)$ to each deletion. Deletions thus
take amortized $O(r\log n)$ time.

\subparagraph{Queries.}
We report the $t$-lowest planes at a query point $q$ as follows. We consider the cutting for which $k_i = 2^i \lceil n/r \rceil = O(t)$, so at level $i = \lceil \log(Ctr/n) \rceil$, for some large enough constant $C$. When $t < n/r$, there is no such cutting, so we query the lowest level cutting instead. We find the prism containing $q$ by a point location query. As the largest cutting has size $O(r)$, this takes $O(\log r)$ time. We then simply report the $t$ lowest planes at $q$ by going through the entire conflict list. This results in a query time of $O(\log r + n/r + t)$.

\subparagraph{Reducing space usage.} When $n$ is large w.r.t. $r$, that is $n \geq r^{1/\eps}$, we can use a similar approach to Chan~\cite{Chan10, Chan12-kNN} to achieve linear space usage. Instead of storing the conflict lists explicitly, we only store the prisms of the shallow cuttings. Additionally, we store the planes in an auxiliary halfspace range reporting data structure~\cite{Agarwal95} with $O(n\log n)$ preprocessing time, $O(n^{1-\eps})$ query time, $O(\log n)$ deletion time, and linear space. This results in a linear space data structure. To delete a plane, we simply delete it in the auxiliary halfspace range reporting data structure in $O(\log n)$ time. After $\frac{n}{2r}$ deletions we rebuild the entire data structure. Thus the amortized deletion time remains $O(r \log n)$.

When performing a query, we first locate the prism $\nabla$ as before. We then query the halfspace range reporting data structure with the intersection point of the vertical line through $q$ and the roof of $\nabla$. We find the $t$ lowest planes by going through the $O(t)$ returned planes. This results in a query time of $O(n^{1-\eps} + \log r + n/r + t)$. Because $n \geq r^{1/\eps}$, we have $n^{1-\eps} \leq n/r$, thus the query time is $O(\log r + n/r + t)$.
\end{proof}

\subparagraph{General data structure.}
The data structure of Lemma~\ref{lem:bootstrapping_ds} can be applied to any type of functions for which we have an algorithm to compute vertical $k$-shallow cuttings. Refer to Section~\ref{sub:Shallow_cuttings} for the definitions regarding the computation of $k$-shallow cuttings. The data structure uses $O(S(n,1)\log r)$
space. Note that the ``lowest'' cutting we use is an $n/r$-shallow
cutting. It follows that constructing all shallow
cuttings takes $O(T(n,n/r)\log r)$ time. To delete a function, we remove all occurrences of the function from the conflict lists in 
\begin{equation*}
\sum_{i=0}^{\log r} S(n,k_i) = \sum_{i=0}^{\log r} \frac{1}{k_i}f(n) = \sum_{i=0}^{\log r} \frac{r}{2^i}f(n) = O((r/n)S(n,1))
\end{equation*}
time, where we used the assumption that $S(n,k) = \frac{1}{k}f(n)$, for some function $f(n)$. Additionally, we charge $O((r/n)T(n,n/r)\log r)$ to the deletion to pay for the global rebuild. To answer a query, we
simply find the prism containing $q$ in one cutting, so the query time is $O(Q(n,n/r) + n/r +
t)$. This results in the following general lemma.

\begin{lemma}\label{lem:bootstrapping_ds_general}
    For any fixed $r$, we can construct a deletion-only data structure of size $O(S(n,1) \log r)$ in $O(T(n,n/r)\log r)$ time that stores a set of $n$ functions, allows for $t$-lowest functions queries in $O(Q(n,n/r) + n/r + t)$ time and deletions in $O((r/n)(S(n,1) + T(n,n/r)\log r))$ amortized time. 
\end{lemma}

\subsection{A general dynamic $k$-NN data structure}
\label{sub:improved-k-NN}

In Section~\ref{sub:improved-k-NN} and~\ref{sub:general_dynamic_k-nn}, we generalize the dynamic $k$-NN data structure from Section~\ref{sub:Chan-k-NN} to support other types of distance functions.
First, in this section, we analyze the space usage and running times of the data structure for planes, as this is somewhat easier to follow. Next, in Section \ref{sub:general_dynamic_k-nn}, we analyze these for the general data structure.
\begin{lemma}
    There is a fully dynamic data structure of size $O(n)$ that stores a set of $n$ point sites and allows for planar $k$-nearest neighbors queries using the Euclidean distance in $O(\log^2n/\log\log n + k)$ time, insertions in $O(\log^{3+\eps}n/\log\log n)$ amortized time, and deletions in $O(\log^{5 + \eps} n/\log\log n)$ amortized time.
\end{lemma}

\begin{proof}
To generalize the dynamic $k$-NN data structure from Section~\ref{sub:Chan-k-NN} to other types of distance functions, we replace the deletion-only data structure used by Chan~\cite{Chan12-kNN} as $\D_0$ data structure by the data structure of Section~\ref{sub:Bootstrapping}. In our approach, we use the same data structure for both small and large conflict lists. Queries and updates are performed as before (see Sections~\ref{sub:dynamic-1NN} and~\ref{sub:Chan-k-NN}). This results in a dynamic $k$-lowest functions data structure that can be used for any type of functions for which we can construct $k$-shallow cuttings. As our $\D_0$ data structure only uses $k$-shallow cuttings, our approach is also somewhat simpler than Chan's, albeit at a slight increase in space usage. This problem can be solved by using the space saving idea discussed in Section~\ref{sub:Bootstrapping}.

\subparagraph{Query time.} 
Our bootstrapping data structure $\D_0$ has query time $Q_0(n') = O(\log r + n'/r)$ and deletion time $D_0(n') = O(r\log n')$. Because we query the shallow cutting at level $j_k$, the size of each conflict list we query is $O(k)$. By again using our scheme to find the $k$-nearest neighbors over the substructures simultaneously, the query time becomes:
\begin{align*}
    Q(n) &= O([\log n + Q_0(O(k))] \log_b n + k) \\
    &= O((\log n + (\log r + k/r))\log_b n + k).
\end{align*}
If we set $r = \log n$ (and $b = \log^\eps n$ just like Chan) we get
\begin{align*}
    Q(n) &= O((\log n + \log \log n + k/\log n)\log n/\log\log n + k) \\
    &= O(\log^2n/ \log\log n + k).
\end{align*}
Thus using our $\D_0$ data structure does not affect the query time.

\subparagraph{Update time.}
In the following we analyse the update time of the data structure in more detail. The update time given by Chan is: \begin{equation*}
U(n) = O(b^{O(1)} \log^6 n + \max_{m\leq n} D_0(m) \cdot b^{O(1)}\log^5 n).
\end{equation*}
Note that this update time is based on the old approach of Chan, where a plane can occur $O(b\log^2n)$ times in a tower. To give a more detailed analysis, we first study the insertion time and then the deletion time of $\D$. 

Lemma \ref{lem:properties} states that insertion time is given by $I(n) = O(b\log_b n \cdot (P(n)/n))$, where $P(n)$ is the preprocessing time of $\D$. Our
preprocessing time increases w.r.t. the original data structure, since after building the hierarchy of shallow cuttings for a tower, we additionally need to build the structures $\D_0$ on each of the conflict lists. As before, building the shallow cuttings takes $O(n\log n)$ time \cite{Chan19}. Next, we analyse the time to build all data structures $\D_0$. Note that the cutting at level $j$ in the hierarchy consists of $O(n/k_j)$ prisms, and the size each conflict list in the cutting is $O(k_j)$. Also note, there are $\log (n/k_0)$ cuttings in the hierarchy, as the lowest level cutting is a $k_0$-shallow cutting. Let $\alpha$ be the constant bounding the size of the conflict lists. Using that $P_0(n') = O(n' \log n')$, we find the following running time for building all $\D_0$ data structures of a single tower: 
\begin{equation*}
    \begin{split}
         \sum_{j=0}^{\log\frac{n}{k_0}} O\left(\frac{n}{ k_j}\right) \cdot P_0(\alpha k_j) &= \sum_{j=0}^{\log\frac{n}{k_0}} O\left(\frac{n}{ k_j}\right) \cdot O\left(\alpha k_j \log(\alpha k_j)\right) \\
         &= \sum_{j=0}^{\log\frac{n}{k_0}} O\left(n \log(\alpha k_j)\right) \\
         &= O\left(\log(n/k_0) \cdot n \log\left(\alpha k_0 2^{\log(n/k_0)}\right) \right) \\
         &= O(n\log^2n).
    \end{split}
\end{equation*}
The preprocessing time thus adheres to the recurrence relation $P(n) \leq P(n/b) + O(n \log^2 n)$, which solves to $P(n) = O(n \log^2n)$. It follows that $I(n) = O(b\log_b n \cdot (P(n)/n)) = O(b\log^2 n \log_b n) = O(\log^{3+\eps}n/\log\log n)$. Note that the improvement of building all shallow cuttings in a tower in $O(n\log n)$ time does not improve the insertion time to $O(\log^2 n)$ as in the $1$-NN data structure, because building the $\D_0$ data structures is the dominant term.

When deleting a plane $h$ that is live in tower $\Ti$, we remove the plane from 
all $\D_0$ of $\Ti$ with $h \in \D_0$. There are at most $O(b\log n)$ such data structures $\D_0$, because a plane can occur at most $O(b\log n)$ times in a tower. On the other hand, by Lemma \ref{lem:properties}, deleting a plane causes amortized $O(b \log n \log_b n)$ reinsertions. Each reinserted plane is also removed from the structures $\D_0$ of a single tower. This results in amortized $O(b \log_b n \log n)$ planes that are again removed from $O(b\log n)$ structures. We can thus formulate the deletion time as $D(n) = O(b \log n\log_b n \cdot (b\log n \cdot D_0(n) + I(n)))$. Plugging in $D_0(n) =
O(\log^2n)$ and $I(n) = O(b \log^2 n \log_b n)$, we find $D(n) = O(b^2\log^4 n \log_b n + b^2\log^3n \log_b^2
n) = O(b^{O(1)} \log^4n \log_b n) = O(\log^{5+\eps} n / \log \log n)$.

\subparagraph{Space usage.} The space usage of a $\D_0$ data structure storing $n'$ planes is $O(n'\log r)$. The space usage of $\D$ is thus $S(n) = \sum_{j=0}^{\log\frac{n}{k_0}} \frac{n}{k_j} \cdot O(k_j\log r) = O(n\log n\log\log n)$. Using the space reduction idea mentioned in Section~\ref{sub:Bootstrapping}, the space of a $\D_0$ data structure is only $O(n')$ for $n'\geq r^{1/\eps}$. The space usage of the entire data structure is then
\begin{equation*}
    S(n) = \sum_{j=0}^{j'} \frac{n}{k_j} \cdot O(k_j\log r) + \sum_{j=j'}^{\log\frac{n}{k_0}} \frac{n}{k_j} \cdot O(k_j) = O(j'n\log r + n\log n),
\end{equation*}
where $j'$ is such that $k_{j'} = r^{1/\eps}$. Using that $r = \log n$ we find $S(n) = O(n\log n)$.
\end{proof}

\subsection{Analysis of the running times for general shallow cuttings}
\label{sub:general_dynamic_k-nn}
We can use a similar scheme in a more general setting. In this section,
we describe the running times of the data structure w.r.t. any
algorithm that can construct vertical $k$-shallow cuttings. We assume that
$S(n,k) = \frac{1}{k}f(n)$, for some function $f(n)$. This implies
that the size of a $k$-shallow cutting including its conflict lists is
$S(n,k) \cdot O(k) = O(S(n,1))$.

We will prove the following theorem, which is similar to Lemma \ref{lem:kaplan}, but for $k$-lowest functions queries.

\thmkNNgeneral*

\begin{proof}
As in Section~\ref{sub:improved-k-NN}, we apply the $\D_0$ data structure from Section~\ref{sub:Bootstrapping} to the data structure described in Sections~\ref{sub:dynamic-1NN} and~\ref{sub:Chan-k-NN}. Instead of Lemma~\ref{lem:bootstrapping_ds}, we apply the general data structure of Lemma~\ref{lem:bootstrapping_ds_general}.

\subparagraph{Query time.} 
Using the query time of the general $\D_0$ data structure of Lemma~\ref{lem:bootstrapping_ds_general}, we have $Q_0(n') = O(Q(n',n'/r) + n'/r)$. Just as in the Euclidean case, we use $r = \log n$ and $b = \log^{\eps}n$. The query time then becomes
\begin{equation*}
    \begin{split}
        Q(n) &= O([Q(n,1) + Q_0(O(k))] \log_b n + k) 
        \\&= O([Q(n,1) + Q(k,k/r) + k/r]\log_b n + k) \\
        &= O(Q(n,1)\log n/\log\log n + k/\log\log n + k) 
        \\&= O(Q(n,1)\log n/\log\log n + k).
    \end{split}
\end{equation*}
Note that $Q(k,k/r) = O(Q(n,1))$, as $Q(n,k)$ is non-decreasing in $n$. Also, $O(\log_b n/ r) = O((\log n/ \log \log n) /\log n) = O(1/\log \log n)$.

\subparagraph{Update time.}
To determine the insertion time, we start by analyzing the preprocessing time of the data structure. In the preprocessing, we first construct the $O(\log n)$ shallow cuttings of a single tower in $O((T(n,1) \log n)$ time. Then, we again build all structures $\D_0$ on (the live part of) the resulting conflict lists. A $k_j$-shallow cutting consists of $S(n,k_j)$ prisms, each with conflict list of size $\leq \alpha k_j$. The time to build the $\D_0$ data structures is thus:
\begin{equation*}
\begin{split}
    \sum_{j = 0}^{\log\frac{n}{k_0}} S(n,k_j) \cdot P_0(\alpha k_j) &=
    \sum_{j = 0}^{\log\frac{n}{k_0}} S(n,k_j) \cdot O(T(\alpha k_j,\alpha k_j/r)\log r) \\
    &= O\left(S(n,1)\log r \cdot  \sum_{j = 0}^{\log\frac{n}{k_0}} \frac{1}{k_j} T(\alpha k_j,\alpha k_j/r)\right) \\
    &= O\left(S(n,1)\log r \cdot  \sum_{j = 0}^{\log\frac{n}{k_0}}
      \frac{1}{k_j} T(\alpha k_j,1)\right) \\
    &= O\left((S(n,1)T(n,1)/n)\log n  \log r \right) \\
    &= O\left((S(n,1)T(n,1)/n)\log n  \log \log n \right).
\end{split}
\end{equation*}
Here we used that $T(\alpha k_j,1)\geq \alpha k_j$ and $T(n,k)$ is non-decreasing in $n$, which implies that $T(\alpha k_j,1)/k_j \leq T(n,1)/n$.
It follows that 
\begin{equation*}
    P(n) = O((S(n,1)T(n,1)/n)\log n \log \log n),
\end{equation*}
and by Lemma~\ref{lem:properties}:
\begin{equation*}
    I(n) = O(b \log_b n \cdot P(n)/n) = O((S(n,1)T(n,1)/n^2)\log^{2+\eps}n).
\end{equation*}
Note that the insertion time has increased by an  $O(\log\log n)$ factor compared to the Euclidean case, because the preprocessing time of the $\D_0$ data structure now depends on $\log r$.

The deletion time is still determined by the number of reinsertions caused by a deletion. According to Lemma \ref{lem:properties}, the amortized number of reinsertions caused by a deletion is $O((S(n,1)/n)b\log n\log_b n)$. Each reinsertion in turn causes $O((S(n,1)/n)b\log n)$ updates on the $\D_0$ data structures. By Lemma~\ref{lem:bootstrapping_ds_general} we have that $D_0(n) = O((r/n)(S(n,1) + T(n,n/r)\log r))$. So, the deletion time is given by:
\begin{equation*}
\begin{split}
    D(n) &= O((S(n,1)/n)b\log n\log_b n  (D_0(n) (S(n,1)/n)b\log n + I(n))) \\
    &= O((S(n,1)/n)\log^{2+\eps} n/\log\log n  (D_0(n)  (S(n,1)/n)\log^{1+\eps} n + I(n)))\\
    &= O((S(n,1)^2T(n,n/\log n)/n^3)\log^{4+\eps} n \\
    &\hspace{13.6em}+ S(n,1)^2 T(n,1) \log^{4+\eps}n /n^3\log \log n) \\
    &= O((T(n,1)/\log \log n + T(n,n/\log n)) \cdot S(n,1)^2\log^{4+\eps}n /n^3).
\end{split}
\end{equation*}
Here, we used that $S(n,1) \leq T(n,1)$.

\subparagraph{Space usage.}
The space usage of the data structure is
\begin{equation*}
    S(n) = \sum_{j=0}^{\log\frac{n}{k_0}} S(n,k_j) \cdot O(S(\alpha k_j,1)\log r) = O(S(n,1)^2/n \cdot \log n\log\log n). \qedhere
\end{equation*}
\end{proof}

\subsection{Applications}\label{sec:fully_dynamic_applications}

\subparagraph{Points in the plane.}

For general distance functions of constant description complexity, for
which the lower envelope of any $t$ functions has $O(t)$ faces, edges,
and vertices, Liu~\cite{Liu20} recently presented an algorithm to
compute $k$-shallow cuttings of linear size. This improves the shallow
cuttings used by Kaplan \etal~\cite{Kaplan17} in the dynamic 1-NN data
structure. Liu~\cite{Liu20} also shows that in the static setting we can use a circular range query data structure~\cite{Agarwal13-circular-range-query} to improve space usage. Like in the Euclidean case, we can use this auxiliary data structure to improve the space usage of the $\D_0$ data structure to linear for large $n$.

\begin{corollary}
There is a fully dynamic data structure of size $O(n\log n)$ that stores a set of $n$ point sites and allows for $k$-nearest neighbors queries using a fixed constant description complexity distance function in $O(\log^2n/ \log\log n + k)$ time. A site can be inserted in $O(\log^{3+\eps} n/\log \log n)$ expected amortized time and deleted in $O(\log^{5+\eps} n/\log \log n)$ expected amortized time.
\end{corollary}

\begin{proof}

The $k$-shallow cutting algorithm by Liu~\cite{Liu20} produces shallow cuttings of size $S(n,k) = O(n/k)$ in $T(n,k) = O(n \log^{3} n\lambda_{s+2}(\log n))$ expected time. A prism containing a query point can be located using a point location data structure in $Q(n,k) = O(\log n)$ time. We can then apply Theorem~\ref{thm:k-NN_general} directly to obtain a $k$-NN data structure.

In the recently published journal version of the paper~\cite{LiuJournal}, Liu improves upon these results even further. The improved algorithm produces shallow cuttings of size $S(n,k) = O(n/k)$ in $T(n,k) = O(n \log n)$ expected time. To be precise, the construction algorithm of Liu~\cite{LiuJournal} can build \emph{all} shallow cuttings in a hierarchy in $O(n\log n)$ expected time. This means that the preprocessing and deletion time of the bootstrapping data structure (Lemma~\ref{lem:bootstrapping_ds_general}) are improved to $O(n\log n)$ and $O(r\log n)$, respectively.
Applying Theorem \ref{thm:k-NN_general} to these shallow cuttings, and adjusting the update times to account for this improvement in the bootstrapping data structure, proves the result.
\end{proof}

\subparagraph{Points in a simple polygon.}
Agarwal \etal~\cite{Staals18} show how to use a data structure based on the structure by Chan to perform geodesic 1-NN queries. We adapt this data structure by again storing the conflict lists in the data structure of Section \ref{sub:Bootstrapping}. The data structure by Agarwal \etal partitions the polygon recursively into two subpolygons, as described in Section \ref{sec:geodesicknn_insertion_only}. It then uses the Chan data structure to find the nearest neighbor in the corresponding subpolygon at each of the $O(\log m)$ levels of the decomposition. Note that in the $k$-NN case, we can again use our scheme of Section \ref{sec:insertion-only} to find the $k$-NN spread over $O(\log m)$ data structures simultaneously.

The shallow-cuttings that are described  in~\cite{Staals18} use pseudo-prisms that are no longer of constant complexity. We refer to the paper for the precise definition. For our purposes, we need only the property that we can still find the prism containing a query point efficiently.

\begin{corollary}
Let $\mathcal{P}$ be a polygon with $m$ vertices. There is a fully dynamic data structure of size $O(n\log^5 n \log m \log\log n + m)$ that stores a set of $n$ point sites in $\mathcal{P}$ and allows for geodesic $k$-nearest neighbors queries in $O(\log^2n\log^2m / \log\log n + k\log m)$ time, insertions in $O(\log^{8+\eps}n\log m + \log^{7+\eps}n\log^3m)$ expected amortized time, and deletions in $O((\log^{12+\eps}n\log m + \log^{11+\eps}n\log^3m)/\log \log n)$ expected amortized time.
\end{corollary}

\begin{proof}
Theorem 23 of \cite{Staals18} states that a $k$-shallow cutting $\Lambda_k(F)$ of $F$ of topological complexity $S(n,k) = O((n/k)\log^2 n)$
can be computed in expected time $O((n/k)\log^3 n(\log n + \log^2 m) + n \log^2 m + n \log^3 n \log m)$. Our improvement for the static $k$-NN data structure of Section \ref{sec:insertion-only} also leads to a slight improvement in the construction time of a shallow cutting. In particular, we can build a shallow cutting in $T(n,k) = O((n/k)\log^3 n(\log n + \log^2 m) + n \log^2 m + n \log^2 n \log (n+m))$ expected time. Locating the correct prism takes $Q(n,k) = O(\log n \log m)$ time. By applying Theorem \ref{thm:k-NN_general} we find for $\D$:
\begin{equation*}
    \begin{split}
        &Q(n) = O(\log^2n\log m / \log\log n + k)\\
        &I(n) = O(\log^{8+\eps}n + \log^{7+\eps}n\log^2m)\\
        &D(n) = O((\log^{12+\eps}n + \log^{11+\eps}n\log^2m)/\log \log n)\\
        &S(n) = O(n\log^5n\log\log n)
    \end{split}
\end{equation*}

Using this data structure at each of the $O(\log m)$ levels of the decomposition and applying the scheme from Section \ref{sec:insertion-only} to achieve $O(\polylog(n,m) + k)$ query time proves the result.
\end{proof}

The following theorem summarizes the results of this section, together with the results from Section~\ref{sub:improved-k-NN}.

\begin{theorem}\label{thm:fully_dynamic_knn}
There is a fully dynamic data structure of size $S(n)$ that stores a set of $n$ sites in the plane and allows for $k$-nearest neighbors queries in $Q(n)$ time, insertions in $I(n)$  amortized time, and deletions in $D(n)$ amortized time. For the general and geodesic case, the update times are expected running times. $\P$ is a simple polygon with $m$ vertices.

\centering
\renewcommand{\arraystretch}{1.8}
\hspace{-0.6cm}\begin{tabular}{clll}
      & \hspace{3mm} Euclidean & \hspace{3mm} General & \hspace{3mm} Geodesic in $\P$ \\
      \midrule
     $Q(n)$ & $O\left(\frac{\log^2n}{\log\log n} + k\right)$ & $O\left(\frac{\log^2n}{\log\log n} + k\right)$ & 
         $O\Big(\frac{\log^2n\log^2m}{\log\log n}+ k\log m\Big)$ \\
     $I(n)$ & $O\left(\frac{\log^{3+\eps}n}{\log\log n}\right)$& $O\left(\frac{\log^{3+\eps}n}{\log\log n}\right)$ & $O(\log^{8+\eps}n\log m + \log^{7+\eps}n\log^3m)$\\
     $D(n)$ & $O\left(\frac{\log^{5 + \eps} n}{\log\log n}\right)$& $O\left(\frac{\log^{5 + \eps} n}{\log\log n}\right)$ &
       $O\Big(\frac{\log^{12+\eps}n\log m}{\log \log n} + \frac{\log^{11+\eps}n\log^3m}{\log \log n}\Big)$\\
     $S(n)$ & $O(n\log n)$ & $O(n\log n)$ &$O(n\log^5 n \log m \log\log n + m)$
\end{tabular}
\end{theorem}

\bibliography{bibliography}

\end{document}